\documentclass{elsarticle}
\usepackage[T1]{fontenc}
\usepackage[latin9]{inputenc}
\usepackage{url}
\usepackage{amsmath}
\usepackage{amsthm}
\usepackage{amssymb}
\usepackage{setspace}
\makeatletter

\providecommand{\tabularnewline}{\\}

\theoremstyle{plain}
\newtheorem{thm}{\protect\theoremname}
\theoremstyle{remark}
\newtheorem{rem}[thm]{\protect\remarkname}
\theoremstyle{definition}
\newtheorem{example}[thm]{\protect\examplename}
\theoremstyle{plain}
\newtheorem{prop}[thm]{\protect\propositionname}
\theoremstyle{plain}
\newtheorem{cor}[thm]{\protect\corollaryname}

\makeatother

\providecommand{\corollaryname}{Corollary}
\providecommand{\examplename}{Example}
\providecommand{\propositionname}{Proposition}
\providecommand{\remarkname}{Remark}
\providecommand{\theoremname}{Theorem}

\begin{document}

\begin{frontmatter}{}

\title{Testing equivalence to binary generalized linear models with application
to logistic regression }
\begin{abstract}
We introduce a new equivalence test to show sufficiently good agreement
of observed data with a binary generalized linear model (GLM). The
test statistic is constructed via the minimum distance method. The
test is developed for the important special case where all covariates
are categorical. The critical values can be calculated using an asymptotic
approximation or by means of bootstrapping. The application of the
test to logistic regression is illustrated on two real data sets.
The finite sample performance of the proposed test is studied by simulations
which are based on these two data sets.
\end{abstract}

\author{Vladimir Ostrovski}

\address{ERGO Group AG, ERGO-Platz 1, 40477 Düsseldorf}

\ead{vladimir\_ostrovski@web.de}

\end{frontmatter}{}

\section{Introduction}

Binary GLMs are widely used in social sciences, marketing, biology,
medical sciences, financial and insurance industry to name a few.
Let $Y$ denote a binary response variable and $X$ be a vector of
covariates. Let $X_{0}=1$ to simplify notation and let $d$ be the
dimension of $X$. A binary GLM states that conditional probability
$P\left(\left.Y=1\right|X=x\right)$ can be approximated by $q_{x}(\beta)=g^{-1}\left(x^{t}\beta\right)$,
where $x\in\mathbb{R}^{d}$ is a value of $X$, $\beta\in\mathbb{R}^{d}$
is a vector of parameters and $g:\left(0,1\right)\rightarrow\left(-\infty,+\infty\right)$
is a bijective function. The function $g$ is called the link function.
The most popular link function is the logit function $g\left(s\right)=\ln\left(\frac{x}{1-x}\right)$,
which corresponds to logistic regression. Another common link function
is the CDF of the standard normal distribution, which corresponds
to probit regression.

A binary GLM can be always fitted to observed data. Therefore, there
is a need to assess the appropriateness of the fitted model, see \citet{Pregibon1981}.
There exists an extensive literature on the goodness of fit tests
for the logistic regression, which are traditionally applied to assess
model quality, see \citet{Hosmer1991}, \citet{HOSMER1997} and \citet{Hosmer2002}.
However, the goodness of fit tests are tailored to establish lack
of fit to observed data. Testing equivalence is an appropriate method
to show sufficiently good agreement of observed data with a binary
GLM. To our best knowledge, there are not any equivalence tests for
the binary GLMs. We restrict ourselves to the case where all covariates
are categorical. This situation is very common in various research
areas. If the categories are not ordered then the category values
should be one-hot encoded. If there is a natural ordinal relationship,
the category values can be assigned corresponding numbers. The observations
can be represented as a multi-way contingency table, where each cell
corresponds to a unique value of covariates.
\begin{rem}
The goodness of fit of a binary GLM is not important in classification
problems, because it is sufficient to be on the correct side of the
decision boundary. Otherwise, binary GLMs are often used to analyse
observed data and to understand the effects of covariates on the outcome.
There are many applications where the modeled conditional probabilities
$\left(q_{x}\left(\beta\right)\right)_{x\in C}$ are used directly
for further processing, see two examples below. A goodness of fit
measure can also be applied to compare different models.
\end{rem}

\begin{example}
The default probability of customer loans is often used for pricing
and risk management. The default probability can be modeled by binary
GLMs based on employment status, gender, marital status, residence
area and other available information. Sufficiently exact modeling
is essential for the correct pricing of default risk.
\end{example}

\begin{example}
The motor vehicle insurance is usually contracted for a period of
one year. The probability of an insurance claim can be modeled by
binary GLMs based on vehicle type, driver experience, residential
area, vehicle use, parking and other available information. The probability
of an insurance claim can be applied for pricing and cost forecasting.
\end{example}

The proposed equivalence test is based on the minimum distance estimation
(MDE) which is also well known to be robust to atypical observations
and model miss-specification, see \citet{Donoho1988}. \citet{Bondell2005}
uses the Cramer-von Mises distance to estimate a logistic regression
model. \citet{Chi2014} proposed to minimize the Euclidean distance
between the observed counting frequencies and corresponding model
probabilities. They derived the L2E loss function, which can be applied
to both continuous and categorical covariates. However, the L2E loss
function is not a distance between observed data and a model. 

Let $C$ be the set of all values of $X$. Let $p_{x}$ denote $P\left(\left.Y=1\right|X=x\right)$
and let $p=\left(p_{x}\right)_{x\in C}$. We consider the minimum
Euclidean distance $d\left(p,\mathcal{M}\right)=\inf_{\beta\in\mathbb{R}^{d}}l_{2}\left(p,q\left(\beta\right)\right)$
between a family of binary GLMs $\mathcal{M}$ and the vector of underlying
conditional probabilities $p$, where $l_{2}$ is the Euclidean distance
and $q\left(\beta\right)=\left(q_{x}(\beta)\right)_{x\in C}$. The
family of binary GLMs $\mathcal{M}$ is specified by the vector of
covariates $X$ and the link function $g$. The Euclidean distance
represents a reasonable trade off between efficiency and robustness,
see \citet{Basu1998}. The equivalence test problem is defined by
$H_{0}=\left\{ d\left(p_{n},\mathcal{M}\right)\geq\varepsilon\right\} $
and $H_{1}=\left\{ d\left(p_{n},\mathcal{M}\right)<\varepsilon\right\} $,
where $\varepsilon>0$ is a tolerance parameter. If $H_{0}$ can be
rejected for a sufficiently small value of $\varepsilon$ then the
corresponding binary GLM is sufficiently close to observed data.

Let $n_{x}$ denote the number of observations under the condition
$X=x$ and let $n=\sum_{x\in C}n_{x}$. We assume that $n_{x}>0$
for any $x\in C$. Let $p_{n,x}$ denote the observed counting frequency
for $X=x$ and let $p_{n}=\left(p_{n,x}\right)_{x\in C}$. The vector
$p_{n}$ is a plug-in estimator of the unknown true conditional probabilities
$p$. If there exists $\beta_{n}\in\mathbb{R}^{d}$ such that $d\left(p_{n},\mathcal{M}\right)=l_{2}\left(p_{n},q\left(\beta_{n}\right)\right)$
then $\beta_{n}$ is called the minimum distance estimator of the
parameter $\beta$. The test statistic for the equivalence test problem
is given by $T\left(p_{n}\right)=\sqrt{n}\left(d^{2}\left(p_{n},\mathcal{M}\right)-\varepsilon^{2}\right)$.
There is no closed formula for the calculation of $T\left(p_{n}\right)$,
so it should be computed numerically.
\begin{rem}
The proposed approach can be extended to the multinomial outcome.
Let $Y\in\left\{ 0,\ldots,k-1\right\} $ and $k>2$ be the number
of the possible outcomes. The conditional distribution of $Y$ given
$X=x$ is multinomial and can be represented by a probability vector
$p_{x}\in\left[0,1\right]^{k}$. A multinomial GLM approximates $p_{x}$
by $q_{x}(\beta)=g^{-1}\left(x\beta\right)$, where $g:\left(0,1\right)\rightarrow\left(-\infty,+\infty\right)$
is a link function, $x\in\mathbb{R}^{d}$ is a value of $X$, $\beta\in\mathbb{R}^{d\times k}$
is a matrix of parameters. The function $g^{-1}$ should be applied
element-wise to the vector $x\beta\in\mathbb{R}^{k}$. 

Let $d$ be some differentiable distance on $\mathbb{R}^{k}$, for
example the Euclidean distance. Let $\overline{d}\left(p,q\left(\beta\right)\right)$
denote the vector $\left(d\left(p_{x},q_{x}\left(\beta\right)\right)\right)_{x\in C}$.
The minimum Euclidean distance between a family of multinomial GLMs
$\mathcal{M}$ and the vector $p=\left(p_{x}\right)_{x\in C}$ of
the multinomial distributions can be defined as $d\left(p,\mathcal{M}\right)=\inf_{\beta}l_{2}\left(\overline{d}\left(p,q\left(\beta\right)\right)\right)$,
where $\beta\in\mathbb{R}^{d\times k}$. The equivalence test problem
and the test statistic are then defined exactly as in the binomial
case. The asymptotic and bootstrap based tests can be derived following
the lines of Section \ref{sec_equivalence_test}.
\end{rem}

\section{Asymptotic distribution and equivalence test\label{sec_equivalence_test}}

First, we derive the asymptotic distribution of the test statistic
$T_{n}$.
\begin{prop}
Let $p_{0}$ be a fixed vector of conditional probabilities so that
$d\left(p_{0},\mathcal{M}\right)=\varepsilon$. Assume that there
exists a continuous function $h:p\mapsto\beta$ on an open neighborhood
$U$ of $p_{0}$ such that $d\left(p,\mathcal{M}\right)=l_{2}\left(p,q\left(h\left(p\right)\right)\right)$
for all $p\in U$. Assume for each $x\in C$ that $\frac{n_{x}}{n}\rightarrow w_{x}$
for $n\rightarrow\infty$ and $w_{x}\in\left(0,1\right)$. Then the
test statistic $T\left(p_{n}\right)$ under $p_{0}$ converges weakly
to $\sum_{x\in C}\frac{2}{\sqrt{w_{x}}}\left(p_{0,x}-q_{x}\left(\beta_{0}\right)\right)Z_{x}$,
where $Z_{x}$ are independently distributed random variables and
$\beta_{0}=h\left(p_{0}\right)$. Each $Z_{x}$ is Gaussian with mean
zero and variance $p_{0,x}\left(1-p_{0,x}\right)$. 
\end{prop}

\begin{proof}
The central limit theorem implies the weak convergence $\sqrt{n_{x}}\left(p_{n,x}-p_{0,x}\right)\rightarrow Z_{x}$.
The random variables $\left(p_{n,x}\right)_{x\in C}$ are independent
and the set $C$ is finite so that $\sqrt{n}\left(p_{n}-p_{0}\right)=\left(\sqrt{\frac{n}{n_{x}}}\sqrt{n_{x}}\left(p_{n,x}-p_{0,x}\right)\right)_{x\in C}\rightarrow\left(\frac{Z_{x}}{\sqrt{w_{x}}}\right)_{x\in C}$
for $n\rightarrow\infty$ by Slutzki's lemma, see \citet[p. 11, Lemma 2.8]{Vaart1998}.
The function $p\mapsto d^{2}\left(p,\mathcal{M}\right)$ is differentiable
at $p_{0}$ with the derivative $\left(2\left(p_{0,x}-q_{x}\left(\beta_{0}\right)\right)\right)_{x\in C}$
by \citet[Theorem 1]{Ostrovski2018}. The assertion follows by the
delta method, see \citet[p. 26, Theorem 3.1]{Vaart1998}.
\end{proof}
\begin{rem}
The assumption, that a continuous minimizer $h$ exists on an open
neighborhood of $p_{0}$, is essential for the numerical calculation
of $d\left(p_{0},\mathcal{M}\right)$ and also for the differentiability
of $p\mapsto d^{2}\left(p,\mathcal{M}\right)$ at $p_{0}$. For the
rest of the paper, we assume that this requirement is fulfilled. This
assumption can be validated numerically by using different starting
points for the optimization.
\end{rem}

\begin{cor}
The asymptotic distribution of $T\left(p_{n}\right)$ under $p_{0}$
is Gaussian with mean zero and variance: 
\begin{eqnarray}
\sigma^{2}\left(p_{0}\right) & = & 4\sum_{x\in C}\frac{1}{w_{x}}p_{0,x}\left(1-p_{0,x}\right)\left(p_{0,x}-q_{x}\left(\beta_{0}\right)\right)^{2}\label{eq_variance}
\end{eqnarray}
\end{cor}

The variance $\sigma^{2}\left(p_{0}\right)$ is a continuous function
of $p_{0}$. Therefore $\sigma^{2}\left(p_{n}\right)$ is a consistent
estimator of $\sigma\left(p_{0}\right)$ by the continuous mapping
theorem. The asymptotic test rejects $H_{0}$ if $T\left(p_{n}\right)\leq c_{\alpha}\sigma\left(p_{n}\right)$,
where $c_{\alpha}$ denotes the lower $\alpha$-quantile of the standard
normal distribution. The minimum tolerance parameter $\varepsilon_{\min}$,
for which the asymptotic test can reject $H_{0}$, equals 
\begin{equation}
\varepsilon_{\min}=\sqrt{d^{2}\left(p_{n},\mathcal{M}\right)-n^{-\frac{1}{2}}c_{\alpha}\sigma\left(p_{n}\right)}\label{min_eps}
\end{equation}
The asymptotic test can be carried out as follows:
\begin{enumerate}
\item Given are the counting frequencies $p_{n}$, the number of observations
$n_{x}$ for all $x\in C$, the tolerance parameter $\varepsilon$
and the significance level $\alpha$.
\item Compute the minimum distance estimator $\beta_{n}$ using some optimization
method.
\item Calculate the minimum distance $d\left(p_{n},\mathcal{M}\right)=l_{2}\left(p_{n},q\left(\beta_{n}\right)\right)$
and the test statistic $T\left(p_{n}\right)=\sqrt{n}\left(d^{2}\left(p_{n},\mathcal{M}\right)-\varepsilon^{2}\right).$
\item Set $w_{x}=\frac{n_{x}}{n}$ for all $x\in C$ and calculate the asymptotic
variance $\sigma^{2}\left(p_{n}\right)$, see (\ref{eq_variance}).
\item Compute the minimum tolerance parameter $\varepsilon_{\min}$, for
which the asymptotic test can reject $H_{0}$, see (\ref{min_eps}).
\item Reject $H_{0}$ if $\varepsilon_{\min}\leq\varepsilon$.
\end{enumerate}
\begin{rem}
It can be shown that the asymptotic test is locally asymptotically
most powerful, see \citet[Proposition 3]{Ostrovski2016} for the proof.
\end{rem}

The variance of the test statistic may be consistently estimated by
the bootstrap method, see \citet[Section 6]{EfronTibshirani1993}
for details. The asymptotic test, in which the variance of the test
statistic is estimated by bootstrapping, will be referred to as the
asymptotic BV test in the remainder of this paper.

In order to improve the finite sample performance of the equivalence
test, we apply the bootstrap-t method, see \citet[Chapter 23]{Vaart1998}.
The bootstrap samples $\hat{p}_{n,x}$ should be generated from the
binomial distribution $B\left(n_{x},p_{n,x}\right)$ for each $x\in C$.
Let $\hat{c}_{\alpha}$ denote the empirical lower $\alpha$-quantile
of the ``studentised'' test statistic $\frac{\sqrt{n}}{\sigma\left(\hat{p}_{n}\right)}\left(T\left(\hat{p}_{n}\right)-T\left(p_{n}\right)\right)$,
where $\hat{p}_{n}=\left(\hat{p}_{n,x}\right)_{x\in C}$. The empirical
quantile $\hat{c}_{\alpha}$ can be computed by the Monte Carlo method
to any degree of accuracy. The bootstrap-t test rejects $H_{0}$ if
$T\left(p_{n}\right)\leq\sigma\left(p_{n}\right)\hat{c}_{\alpha}$.
The minimum tolerance parameter $\varepsilon_{\min}$, for which the
bootstrap-t test can reject $H_{0}$, is $\sqrt{d^{2}\left(p_{n},\mathcal{M}\right)-n^{-\frac{1}{2}}\hat{c}_{\alpha}\sigma\left(p_{n}\right)}$.
The bootstrap-t test is consistent by \citet[p. 330, Theorem 23.4 and p. 331, Theorem 23.5]{Vaart1998}.

\section{Simulation study for logistic regression}

The simulation study is performed for the logistic regression model
because it is the most common type of the binary GLMs. The proposed
equivalence tests are implemented in R and are freely available under
\url{https://github.com/TestingEquivalence/MDLogisticRegressionR}.
All simulations are performed in R-Studio on a scientific workstation.
The simulation study is based on two real data sets: Fiji fertility
survey (Fiji) from \citet{Little1978} and survival of passengers
on the Titanic (Titanic) from \citet{Dawson1995}. Both data sets
are included in the source code. All categories are one-hot encoded. 

The Fiji data set contains the distribution of 1607 interviewed women,
classified by current age (under 25, 25-29, 30-39, 40-49), level of
education (low, high), desire for more children (yes, no) and contraceptive
use (yes, no). In our analysis, contraceptive use is the response
variable. Age, level of education and desire for more children are
covariates. 

The Titanic data set provides information on the fate of 2201 passengers
on the fatal voyage of the ocean liner \textquoteleft Titanic\textquoteright ,
summarized according to class (1st, 2nd, 3rd, Crew), sex (male, female),
age (child, adult) and survival (yes, no). We consider survival as
a response variable. The covariate age is omitted because there are
no children in 1st class, 2nd class and crew. The covariates class
and sex are used for the logistic regression.

The minimum number of observations $\min_{x\in C}n_{x}$ should be
sufficiently large so that the observed counting frequencies $p_{n,x}$
have appropriate quality for all $x\in C$. The considered data sets
meet this requirement because $\min_{x\in C}n_{x}=14$ for the Fiji
data set and $\min_{x\in C}n_{x}=23$ for the Titanic data set. 

\subsection{Estimation of regression coefficients}

MLE and MDE of regression coefficients are compared in Table \ref{tab_reg_coef}.
Additionally, the Euclidean distances $l_{2}\left(p_{n},q\left(\beta\right)\right)$
between observed data and the corresponding logistic regression models
are shown. 

In case of the Fiji data set, the distances $l_{2}\left(p_{n},q\left(\beta\right)\right)$
for MLE and MDE are close to each other. MLE and MDE of the regression
coefficients are also quite similar. MDE of most coefficients have
considerably higher standard deviations than MLE. The standard deviation
of the distance $l_{2}\left(p_{n},q\left(\beta\right)\right)$ is
almost the same for MDE and MLE, so that the logistic regression model
estimated by MDE is not more volatile overall.

In case of the Titanic data set, the model estimated by MDE is substantially
closer to the observed data compared to the model estimated by MLE.
Some regression coefficients differ significantly depending on the
estimation method. The standard deviations of MDE regression coefficients
are higher than those of MLE. The standard deviation of the distance
$l_{2}\left(p_{n},q\left(\beta\right)\right)$ is significantly lower
for MDE compared to MLE. Hence, the logistic regression model estimated
by MDE can be considered much more robust.

\begin{table}[h]
\begin{tabular}{cccccc}
\hline 
Data set & Parameter & MLE & MDE & SD MLE & SD MDE\tabularnewline
\hline 
Fiji & $l_{2}\left(p_{n},q\left(\beta\right)\right)$  & 0.32 & 0.30 & 0.066 & 0.066\tabularnewline
 & intercept  & -0.81 & -0.96 & 0.16 & 0.27\tabularnewline
 & wants more children yes  & -0.83 & -0.95 & 0.12 & 0.17\tabularnewline
 & education low  & -0.32 & -0.41 & 0.13 & 0.18\tabularnewline
 & age 25-29 & 0.39 & 0.53 & 0.18 & 0.35\tabularnewline
 & age 30-39 & 0.91 & 1.16 & 0.17 & 0.31\tabularnewline
 &  age 40-49 & 1.19 & 1.52 & 0.21 & 0.34\tabularnewline
\hline 
Titanic & $l_{2}\left(p_{n},q\left(\beta\right)\right)$  & 0.26 & 0.15 & 0.036 & 0.020\tabularnewline
 & intercept & -0.35 & -0.61 & 0.12 & 0.15\tabularnewline
 & class 2nd & -0.95 & -1.00 & 0.16 & 0.22\tabularnewline
 & class 3rd & -1.66 & -2.82 & 0.16 & 0.37\tabularnewline
 & class crew & -0.88 & -0.69 & 0.15 & 0.22\tabularnewline
 & gender female & 2.42 & 3.33 & 0.13 & 0.32\tabularnewline
\hline 
\end{tabular}

\caption{Comparison of the regression coefficients, which are estimated by
MLE and MDE. The Euclidean distance between observed data and the
corresponding logistic regression model is denoted $l_{2}\left(p_{n},q\left(\beta\right)\right)$.
The standard deviations (SD) are estimated by bootstrap method, see
Section \ref{sec_equivalence_test} for the resampling scheme.\label{tab_reg_coef}}
\end{table}

\subsection{Test size and test results}

An appropriate value of the tolerance parameter $\varepsilon$ can
be found based on the test power under the assumption that the logistic
regression model is true. Table \ref{tab_test_size} shows the tolerance
parameter $\varepsilon$ as a function of the test power at the regression
model estimated by MDE. Let $\beta$ denote the vector of model parameters.
The test power is calculated by the Monte Carlo method, where the
samples are generated from the binomial distribution $B\left(n_{x},q_{x}\left(\beta\right)\right)$
for each $x\in C$. 

We set the parameter $\varepsilon$ so that the test power equals
0.9 at the logistic regression model, see Table \ref{tab_test_size}.
 Thus, we control the type II error rate by setting $\varepsilon$
sufficiently large. In case of the Titanic data set, all three tests
perform very similarly. $H_{0}$ cannot be rejected by any of the
considered tests because $\varepsilon_{\min}$ is larger than the
corresponding values of the parameter $\varepsilon$, see Table \ref{tab_test_size}.
Therefore, the logistic regression model of the Titanic data set is
not appropriate for the observed data. In case of the Fiji data set,
$H_{0}$ can not be rejected by any of the considered tests because
$\varepsilon_{\min}>\varepsilon$ in all cases. Therefore, the logistic
regression model of the Fiji data set is not equivalent to the observed
data.

\begin{table}[h]
\begin{tabular}{cccccccc}
\hline 
Data Set & Test & 0.9 & 0.8 & 0.7 & 0.6 & 0.5 & $\varepsilon_{\min}$\tabularnewline
\hline 
Fiji & asymptotic & 0.36 & 0.33 & 0.31 & 0.29 & 0.27 & 0.40\tabularnewline
 & asymptotic BV & 0.37 & 0.34 & 0.32 & 0.31 & 0.29 & 0.41\tabularnewline
 & bootstrap-t & 0.32 & 0.28 & 0.25 & 0.22 & 0.20 & 0.36\tabularnewline
\hline 
Titanic & asymptotic & 0.13 & 0.12 & 0.10 & 0.10 & 0.09 & 0.18\tabularnewline
 & asymptotic BV & 0.13 & 0.12 & 0.11 & 0.10 & 0.09 & 0.18\tabularnewline
 & bootstrap-t & 0.16 & 0.13 & 0.11 & 0.10 & 0.08 & 0.18\tabularnewline
\hline 
\end{tabular}

\caption{Tolerance parameter $\varepsilon$ as a function of the test power
at the model estimated by MDE. The values of the test power 0.9,...,0.5
are column names. Column $\varepsilon_{\min}$ contains the minimum
tolerance parameter $\varepsilon$, for which $H_{0}$ can be rejected
for the original data set. All tests are at the nominal level 0.05.
The number of simulations is 1000 for each experiment. \label{tab_test_size}}

\end{table}

\subsection{Type I error rates}

In this section, we study Type I error rates of the proposed tests
using Monte Carlo simulations. The boundary of $H_{0}$ is very complex
so that it is difficult to find boundary points that have the highest
rejection probability. Therefore, we consider many randomly generated
boundary points of $H_{0}$ which are based on the two original data
sets. The boundary points are randomly generated for $\varepsilon>d\left(p_{n},\mathcal{M}\right)$
using the following algorithm:
\begin{enumerate}
\item Given are the number of observations $n_{x}$ for all $x\in C$, the
counting frequencies $p_{n}$ and the tolerance parameter $\varepsilon$
so that $\varepsilon>d\left(p_{n},\mathcal{M}\right)$.
\item Draw a random sample $\hat{p}_{n,x}$ from the binomial distribution
$B\left(n_{x},p_{n,x}\right)$ for each $x\in C$. Let $\hat{p}_{n}$
denote $\left(\hat{p}_{n,x}\right)_{x\in C}$.
\item If $d\left(\hat{p}_{n},\mathcal{M}\right)<\varepsilon$ then reject
$\hat{p}_{n}$ and go back to step 2. Otherwise go to the next step.
\item Consider the linear combination $w\mapsto wp_{n}+\left(1-w\right)\hat{p}_{n}$
for $a\in\left[0,1\right]$. Find $w_{n}\in\left[0,1\right]$ so that
$d\left(w_{n}p_{n}+\left(1-w_{n}\right)\hat{p}_{n},\mathcal{M}\right)=\varepsilon$.
The value of $a_{n}$ can be found using any line search method. 
\item Return $w_{n}p_{n}+\left(1-w_{n}\right)\hat{p}_{n}$, which is a random
boundary point of $H_{0}$.
\end{enumerate}
By construction, the randomly generated boundary points are close
to the observed counting frequencies $p_{n}$ if $\varepsilon$ is
close to $d\left(p_{n},\mathcal{M}\right)$. Considering the test
power at 100 boundary points, we shed some light on the Type I error
rates in the neighborhood of the original data sets, see Table \ref{tab_test_power}
for the summary of the simulation results. The values of parameter
$\varepsilon$ in Table \ref{tab_test_power} are chosen so that $\varepsilon>d\left(p_{n},\mathcal{M}\right)$
and the test power is over 0.8 at the logistic regression model estimated
by MDE. 

The test power varies considerably from point to point. The average
test power of the asymptotic test and also of the asymptotic BV test
is below the nominal level 0.05. The average test power of the bootstrap-t
test is considerably larger than the nominal level 0.05 so that the
bootstrap-t test is not conservative. The bootstrap-t test fails to
control the type I error rate because the maximum test power is far
above the nominal level for both data sets. The asymptotic test shows
some non-conservative tendency because the maximum test power for
the Titanic data set is somewhat too large. The asymptotic BV test
performs well for both examples.

Overall, the asymptotic BV test has the best performance and should
be used. We recommend to supplement the application of the proposed
equivalence tests with a simulation of the test power at the estimated
model to get insight into the appropriate values of the tolerance
parameter $\varepsilon$. 

\begin{table}[h]
\begin{tabular}{cccccc}
\hline 
Data Set & Test & $\varepsilon$ & Mean & Max & SD\tabularnewline
\hline 
Fiji & asymptotic & 0.35 & 0.022 & 0.056 & 0.009\tabularnewline
 & asymptotic BV & 0.35 & 0.013 & 0.042 & 0.006\tabularnewline
 & bootstrap-t & 0.35 & 0.081 & 0.115 & 0.013\tabularnewline
\hline 
Titanic & asymptotic & 0.16 & 0.039 & 0.082 & 0.011\tabularnewline
 & asymptotic BV & 0.16 & 0.028 & 0.067 & 0.009\tabularnewline
 & bootstrap-t & 0.16 & 0.061 & 0.103 & 0.012\tabularnewline
\hline 
\end{tabular}

\caption{Summary of the simulated test power at 100 randomly generated boundary
points of $H_{0}$. All tests are at the nominal level 0.05. The number
of simulations is 1000 for each experiment. \label{tab_test_power}}
\end{table}

\bibliographystyle{plainnat}
\bibliography{VO_literature}

\end{document}